\documentclass[twoside]{article}

\usepackage[sc]{mathpazo} 
\usepackage[T1]{fontenc} 
\usepackage{microtype} 

\usepackage[hmarginratio=1:1,top=32mm,columnsep=20pt]{geometry} 
\usepackage{multicol} 
\usepackage[hang, small,labelfont=bf,up,textfont=it,up]{caption} 
\usepackage{booktabs} 
\usepackage{float} 
\usepackage{hyperref} 

\usepackage{lettrine} 
\usepackage{paralist} 

\usepackage{mathrsfs}
\usepackage{amsmath}
\usepackage{amsfonts}
\usepackage{mathtools}
\usepackage{amsthm}

\usepackage{abstract} 

\usepackage{titlesec} 
\renewcommand\thesection{\Roman{section}} 
\renewcommand\thesubsection{\Roman{subsection}} 
\titleformat{\section}[block]{\large\scshape\centering}{\thesection.}{1em}{} 
\titleformat{\subsection}[block]{\large}{\thesubsection.}{1em}{} 

\usepackage{fancyhdr} 
\title{\vspace{-15mm}\fontsize{20pt}{10pt}\selectfont\textbf{Arbitrage-free exchange rate ensembles\\over a general trade network}} 

\author{
\textsc{Stan Palasek}\\[2mm] 
\normalsize Princeton University \\ 
\normalsize \href{mailto:spalasek@princeton.edu}{spalasek@princeton.edu} 
}
\date{June 2014}

\newtheorem{theorem}{Theorem}
\newtheorem{definition}[theorem]{Definition}
\newtheorem{lemma}[theorem]{Lemma}
\newtheorem{corollary}[theorem]{Corollary}

\begin{document}

\maketitle 

\thispagestyle{fancy} 


\begin{abstract}

It is assumed that under suitable economic and information-theoretic conditions, market exchange rates are free from arbitrage. Commodity markets in which trades occur over a complete graph are shown to be trivial. We therefore examine the vector space of no-arbitrage exchange rate ensembles over an arbitrary connected undirected graph. Consideration is given for the minimal information for determination of an exchange rate ensemble. We conclude with a topical discussion of exchanges in which our analyses may be relevant, including the emergent but highly-regulated (and therefore not a complete graph) market for digital currencies.

\end{abstract}

\begin{multicols}{2}

\section{Introduction}

Consider a set of goods with a fixed exchange rate defined for each pair traded among rational investors. Provided that there is complete information and the goods have objective worth (as might, for instance, a currency), it is necessary that the rates be such that no market participant can make a strict profit by executing a trade which ends in the same denomination with which it began. Otherwise, because the quantity of each good is conserved and the investors by assumption have identical preferences, executing the trade would leave the others taking an effective loss. We will refer to this assumption as the ``no-arbitrage'' condition.~\cite[p.~322]{harrison} Although history shows that it \textit{is} at times possible to profit from differential cross and direct exchange rates due to, for instance, asymmetric information,~\cite{mavrides, mwangi} we will for now neglect these uncommon complications.

Ellerman notes that the arbitrage-free property of such an exchange is equivalent both to the system being path-independent and to the exchange rates taking on their trivial cross-rate values.~\cite{ellerman} These truths will become evident given the formulation presented here, though we will proceed in the manner of Mirowski and consider only undirected networks in which a reciprocal exists for each potential trade.~\cite{mirowski07} Under the presumption that all relevant economies are free from arbitrage, we will further examine the structure of the space of plausible exchange ensembles from a graph-theoretic perspective. Of particular interest are markets without a universal currency (ie.\ graphs without a vertex connected to all others). In such an economy it is meaningless to speak of ``price'' since no denomination is universal (although we will prove a theorem that provides an equivalent alternative). The resulting ambiguity over the ``value'' of a commodity is of interest to postmodern philosophers who have themselves studied this system.~\cite{mirowski91}

\section{The space of arbitrage-free ensembles}

\begin{definition}\label{exchange}
Let $G$ be a connected $n$th-order graph with vertices $g_1,g_2,\ldots,g_n$ which represents the exchanges between the $n$ goods that may occur. Define an exchange matrix associated with $G$ as any $n\times n$ matrix $\tilde{E}$ of positive reals where $\tilde{E}_{i,j}$ is the exchange rate of $g_j$ to $g_i$ or equivalently the price of $g_i$ in units of $g_j$ with the additional property that if $(i,j)\notin G$, then $E_{i,j}=1$. Furthermore, for an exchange matrix $\tilde E$, define its additive exchange matrix $E$ to be the $n\times n$ matrix of reals with $E_{i,j}=\ln\tilde{E}_{i,j}$.
\end{definition}

\begin{definition}\label{noarbitrage}
Define a no-arbitrage matrix over a connected graph $G$ as an exchange matrix $\tilde{E}$ such that for any closed walk upon $G$ over the vertices $g_{n_1},g_{n_2},\ldots,g_{n_t},g_{n_1}$, we have
\begin{align}\label{mult}
\tilde{E}_{n_t,n_1}\cdot\prod_{i=1}^{t-1}\tilde{E}_{n_i,n_{i+1}}=1
\end{align}
or equivalently for the additive exchange matrix,
\begin{align}\label{add}
E_{n_t,n_1}+\sum_{i=1}^{t-1}E_{n_i,n_{i+1}}=0.
\end{align}
\end{definition}
Note that Definition \ref{noarbitrage} considers \eqref{mult} and \eqref{add} over \textit{all} closed walks on $G$, even those with repeated vertices (not counting the repetition of $g_{i_1}$). However, because a closed walk with repeated vertices can trivially be fragmented into a sequence of closed walks without repeated vertices and \eqref{mult} and \eqref{add} apply associatively, it is equivalent to consider the equations over non-repeating closed walks (``cycles'').

\begin{definition}
Let $\mathscr{M}(G)$ be the set of all additive no-arbitrage matrices over the connected graph $G$.
\end{definition}

\begin{lemma}\label{subspace}
$\mathscr{M}(G)$ is a subspace of the vector space of $n\times n$ matrices over standard matrix operations.
\end{lemma}
\begin{proof}
Choose any $E,E'\in\mathscr{M}(G)$ and $c\in\mathbb{R}$ and consider the matrix $E+c\cdot E'$. Evaluating it in the left-hand side of \eqref{add}, the terms separate due to linearity, again yielding zero. Furthermore, since both matrices are over $G$, they contain zero entries wherever an edge is not in $G$; therefore so does their linear combination. Thus all the conditions are satisfied for $E+c\cdot E'\in\mathscr{M}(G)$ so the space is closed under linear combination.
\end{proof}

\begin{lemma}\label{facts}
Let $G$ be an $n$th-order connected graph and $E$ an additive no-arbitrage matrix on it. Then for any $i,j$,
\begin{enumerate}
\item $E_{i,i}=0$
\item $E_{i,j}=-E_{j,i}$.
\end{enumerate}
\end{lemma}
\begin{proof}
If $(i,j)\notin G$, then by Definition \ref{exchange} $E_{i,j}=0$ and the second statement is satisfied. Otherwise if $(i,j)\in G$, then we can apply Definition \ref{noarbitrage} to the closed walk $(i,j,i)$ to obtain $E_{i,j}+E_{j,i}=0$ and the second statement is again satisfied. Letting $i=j$, the first statement follows from the second.
\end{proof}

Note that if $i$ is not reflexively connected to itself, the first statement of the lemma holds not for any economically relevant reason, but rather because we arbitrarily defined the entries of $E$ corresponding to non-existent edges of $G$ to be 0. This convention is nonetheless useful as it preserves $\mathscr{M}$'s additive and scalar-multiplicative closure which we needed in the proof of Lemma \ref{subspace}.

\begin{definition}\label{basis}
Let $G$ be a connected graph of order $n$ and $1\leq i_k,j_k\leq n$ for $k=1,2,\ldots,t$. A collection of ordered pairs $(i_1,j_1),(i_2,j_2),\ldots,(i_t,j_t)$ is a basis for $G$'s additive exchange matrix if fixing the entries $E_{i_1,j_1},E_{i_2,j_2},\ldots,E_{i_t,j_t}$ of the no-arbitrage additive exchange matrix $E$ of $G$ uniquely and minimally determines the rest of the matrix. We will refer to $t$ as the dimension.
\end{definition}

\begin{lemma}\label{complete}
Let $G$ be the complete\footnote{When we refer to a complete graph here and elsewhere, it is irrelevant whether one includes or excludes the reflexive edges connecting each vertex to itself, ie.\ whether ones or zeros constitute the main diagonal of the adjacency matrix. The main diagonal of the exchange matrix will in either case be zeros due either to no-arbitrage over the reflexive loops or the convention of setting to zero the entries corresponding to absent edges.} graph of order $n$. Then for any $k$, the entries $(k,1),(k,2),\ldots,(k,k-1),(k,k+1),\ldots,(k,n-1),(k,n)$ are a basis for the exchange matrix of $G$.
\end{lemma}
\begin{proof}
Fix values for $E_{k,1}$ through $E_{k,n}$ excluding $E_{k,k}$. Consider $i,j$ satisfying $(i,j)\in G$. Case~1: If $i=j=k$, then $E_{i,j}$ is determined by part 1 of Lemma \ref{facts}. Case~2: Suppose $i=k\neq j$. Then $E_{i,j}=E_{k,j}$ is among the fixed values. Switching the roles of $i$ and $j$ is determined likewise up to a sign by the second part of Lemma \ref{facts}. Case~3: Suppose $i,j\neq k$. Since the graph is complete, $(j,k)$ and $(k,i)$ are in $G$. Therefore $g_i,g_j,g_k,g_i$ form a closed walk and by Definition \ref{noarbitrage}, $E_{i,j}+E_{j,k}+E_{k,i}=0$. Using the second part of Lemma \ref{facts}, this rearranges to $E_{i,j}=E_{k,j}-E_{k,i}$ which is again determined in terms of the fixed values. Thus we've shown that every entry of $E$ is determined uniquely.

Next we illustrate that a smaller collection of entries does not determine the exchange matrix uniquely. Let
\begin{align}
E_{i,j}=E_{k,j}-E_{k,i}
\end{align}
for all $1\leq i,j,k\leq n$. Then \eqref{add} becomes
\begin{align}
E_{k,n_1}-E_{k,n_t}+\sum_{i=1}^{t-1}(E_{k,n_{i+1}}-E_{k,n_i})=0
\end{align}
which telescopes to a tautology. Therefore such an exchange matrix is no-arbitrage regardless of the choice of $E_{k,1},E_{k,2},\ldots,E_{k,k-1},E_{k,k+1},\ldots,E_{k,n-1},E_{k,n}$. We conclude that no smaller set could determine the exchange matrix uniquely; thus both conditions of Definition \ref{basis} are met.
\end{proof}

In the next lemma, unlike before, we will additionally refer to the standard linear algebra definition of basis. However the result will allow us to identify the Definition \ref{basis} understanding of basis and dimension precisely with a basis for and the dimension of the space $\mathscr{M}(G)$.

\begin{definition}\label{eps}
Let $G$ be a connected graph of order $n$ and $(i_1,j_1),(i_2,j_2),\ldots,(i_t,j_t)$ a basis for it. For $k=1,2,\ldots,t$, define $\epsilon_k$ to be the matrix uniquely determined by letting $E_{i_k,j_k}=1$ and setting the rest of the basis entries to zero.
\end{definition}

\begin{lemma}\label{equiv}
Let $G$ be a connected graph of order $n$ with $(i_1,j_1),(i_2,j_2),\ldots,(i_t,j_t)$ and $\epsilon_1,\epsilon_2,\ldots,\epsilon_k$ as in Definition \ref{eps}. Then $\sum_{k=1}^ta_k\epsilon_k$ is the unique no-arbitrage exchange matrix over $G$ with $a_k$ at the entry $(i_k,j_k)$ for all $k=1,2,\ldots,t$.
\end{lemma}
\begin{proof}
It follows from Definition \ref{eps} that the $k$th matrix's $(i_k,j_k)$ entry is one and for $m\neq k$, its $(i_m,j_m)$ entry is zero. In mathematical terms,
\begin{align}
(\epsilon_k)_{i_m,j_m}=\delta_{k,m}.
\end{align}
Therefore,
\begin{align}
\left(\sum_{k=1}^ta_k\epsilon_k\right)_{i_m,j_m}&=\sum_{k=1}^ta_k(\epsilon_k)_{i_m,j_m}\\
&=\sum_{k=1}^ta_k\delta_{k,m}\\
&=a_m.
\end{align}
By Lemma \ref{subspace}, $\sum_{k=1}^ta_k\epsilon_k\in\mathscr{M}(G)$. By Definition \ref{eps}, it is unique.
\end{proof}

\begin{theorem}\label{realbasis}
Let $G$ be a connected graph of order $n$ with $(i_1,j_1),(i_2,j_2),\ldots,(i_t,j_t)$ and $\epsilon_1,\epsilon_2,\ldots,\epsilon_k$ as in Definition \ref{eps}. Then the collection of matrices $\epsilon_1,\epsilon_2,\ldots,\epsilon_t$ form a basis of $\mathscr{M}(G)$.
\end{theorem}
\begin{proof}
We must show that every no-arbitrage additive exchange matrix over $G$ has a unique expression as a linear combination of the $\epsilon_k$. By Lemma \ref{equiv}, this is equivalent to the proposition that every no-arbitrage additive exchange matrix over $G$ can be expressed with a unique choice of constants to fill the entries $(i_1,j_1),(i_2,j_2),\ldots,(i_t,j_t)$. This proposition is precisely the meaning of $(i_1,j_1),(i_2,j_2),\ldots,(i_t,j_t)$ being a basis as given by Definition \ref{basis}.
\end{proof}

\begin{corollary}\label{boundaries}
Let $G$ be the complete graph and $H$ a connected tree (acyclic graph), both of order $n$. Then
\begin{align}
\dim\mathscr{M}(G)=\dim\mathscr{M}(H)=n-1.
\end{align}
\end{corollary}
\begin{proof}
The equality for $G$ follows immediately from Lemma \ref{complete} and Theorem \ref{realbasis}.

By definition of a tree, the only closed walks on $H$ upon which we need to invoke \eqref{add} are those which double back on themselves, ie.\ those of the form $(i_1,i_2,\ldots,i_{k-1},i_k,i_{k-1},\ldots,i_2,i_1)$ for some $k\in\mathbb{N}$. Linearity of \ref{add} implies that it is necessary and sufficient to check this condition with $k=2$, yielding anticommutativity as in Lemma \ref{facts} as the sole restriction. Thus the additive exchange has a degree of freedom for each edge, of which there are $n-1$ for any tree.~\cite[p.~14, Cor.~1.5.3]{diestel} Then by Theorem \ref{realbasis}, $\dim\mathscr{M}(H)=n-1$.
\end{proof}

\begin{theorem}\label{big}
If $G$ is a connected graph, then
\begin{align}
\dim\mathscr{M}(G)=|G|-1.
\end{align}
\end{theorem}
\begin{proof}
Let $n=|G|$, the number of vertices in the graph. $G$ must contain a normal spanning tree~\cite[p.~16, Prop.~1.5.6]{diestel}; call it $H_{n-1}$ in recognition of the fact that it has $n-1$ edges. We will henceforth denote the number of edges of a graph by $\|\cdot\|$. Next, construct a sequence of graphs $H_{n-1},H_n,\ldots,H_{\|G\|-1},H_{\|G\|}$ with the following properties for $n-1\leq i<\|G\|$:
\begin{enumerate}
\item $H_i$ is a spanning subgraph of $H_{i+1}$
\item $\|H_{i+1}\|=\|H_i\|+1$
\item $H_{\|G\|}=G$
\end{enumerate}
In other words, the sequence of graphs from $H_{\|G\|}$ to $H_{n-1}$ is the transformation of $G$ into one of its spanning trees, successively removing edges at each step. Since $H_{n-1}$ is by definition a subgraph of $G$, such a sequence evidently exists. We will proceed by induction over the $H_i$ to show that each of the $\mathscr{M}(H_i)$ and in particular $\mathscr{M}(H_{\|G\|})$ has dimension $n-1$.

The base case in which $i=n-1$ is given by Corollary \ref{boundaries}.

For the inductive step, suppose that for some $i$ satisfying $n-1\leq i<\|G\|$, $\dim\mathscr{M}(H_i)=n-1$. Let the edge $(k,m)$ be the singleton element of $H_{i+1}\backslash H_i$. There are three cases. Case~1: Suppose $k=m$. The only new cycle is the reflexive loop $(k,m)=(k,k)$, the exchange matrix entry over which is trivially determined to be 0. Thus the collection of entries that form a basis of $H_i$ in the Definition \ref{basis} sense likewise form one of $H_{i+1}$. By Theorem \ref{realbasis} they have equal dimension. Case~2: Suppose there is exactly one path from $g_k$ to $g_m$ consisting of distinct vertices, call it $\mathscr{C}$. Then the addition of the edge $(k,m)$ creates four cycles on $H_{i+1}$ that did not exist on $H_i$: $(k,m,k)$, $\mathscr{C}\cup(k)$, and their respective reversals. Analogously to Lemma \ref{facts}, the no-arbitrage condition on $(k,m,k)$ implies
\begin{align}\label{anti}
E_{k,m}=-E_{m,k}.
\end{align}
Thus all the edges in $H_{i+1}$ are anticommutative, so if a cycle satisfies \eqref{add} then its reversal must as well. We therefore see that the only nontrivial new constraint introduced by the edge $(k,m)$ is the no-arbitrage condition on one direction of $\mathscr{C}\cup(k)$. Furthermore, by \eqref{anti}, there is only one new unique variable to be determined.

Consider the equation we referenced imposed by \ref{add} on $\mathscr{C}\cup(k)$. Because $H_i$ is a spanning subgraph of $H_{i+1}$ and a basis of size $n-1$ for $H_i$'s exchange matrix determines each of the variables, that basis determines all of the variables in the no-arbitrage equation except for $E_{k,m}$, which in turn is determined by the equation. We therefore see that $H_{i+1}$ has a basis of size $n-1$ in the Definition \ref{basis} sense, so by Theorem \ref{realbasis}, $\dim\mathscr{M}(H_{i+1})=n-1$. Case~3: Suppose there is more than one path from $g_k$ to $g_m$ consisting of distinct vertices. We will illustrate that additional paths impose no greater restriction than the single one we considered in the second case. Let $\mathscr{C}$ and $\mathscr{C}'$ be two non-identical paths from $g_k$ to $g_m$ on $H_i$. $H_{i+1}$ now has two cycles that we did not consider in Case~2: $\mathscr{C'}\cup(k)$ and its reverse. Again by anticommutativity, the no-arbitrage equations imposed by these paths are equivalent. Furthermore, observe that $\mathscr{C}\cup(-\mathscr{C}')$ is a cycle on $H_i$ where $-\mathscr{C}'$ is the reversal of the path $\mathscr{C}'$. Since the basis for $H_i$ provided for no-arbitrage over this cycle as it does not assume the edge $(k,m)$, its no-arbitrage equation can be assumed.
\begin{align}
\sum_{(p,q)\in\mathscr{C}\cup(-\mathscr{C}')}E_{p,q}=0
\end{align}
We may then split the equation over the two paths.
\begin{align}
\sum_{(p,q)\in\mathscr{C}}E_{p,q}=-\sum_{(p,q)\in-\mathscr{C}'}E_{p,q}
\end{align}
Thus the no-arbitrage condition is equivalent over any two paths so the third case reduces to the second. It follows by induction that $\dim\mathscr{M}(H_i)=n-1$ for $n-1\leq i\leq\|G\|$. Letting $i=\|G\|$ proves the theorem.
\end{proof}

\section{Discussion}

The cited literature's consideration of exchange rates over arbitrary networks (see, for instance,~\cite{ellerman, mirowski91}) seems to rely on the premises both that complete graphs are too trivial to be of interest and that nontrivial underlying networks may exist \textit{ex papyro}. However, the conditions formulated both implicitly and explicitly in the introduction may considerably limit the scope of the class of systems which our analyses may encompass. First, we required that all traders have the same preferences, ie.\ that they demand the same price for a given asset; otherwise we would not have a fixed exchange matrix. Second, in order to obtain a nontrivial (here, non-sparse) graph, trades must occur between a variety of pairs of assets; if there is a single commodity that is behaving as a universal currency, then every cross rate is immediately determined triangularly. The empirical existence of idiosyncratic consumer preferences, investment objectives, and information makes the first condition implausible (American reality television provides an explicit counterexample, see \cite{floro}). Furthermore, the second condition obsolesced along with the bartering system in modern economies. Hence currency markets are the only logical applications where our considerations might be relevant. Neglecting differential preferences for foreign goods, a given currency has the same value across individuals (my euro buys as many goods as your euro). 

Now we must contend only with the potential of the underlying network to be not sparse \textit{enough}. Certainly there are active exchanges between virtually every pair of conventional currencies and, as we showed in Lemma \ref{complete}, such complete no-arbitrage networks are uniquely determined by the set of $n-1$ exchange rates with any given currency. We might, however, turn our attention to unconventional electronic currencies.\footnote{Although the Internal Revenue Service might disagree with this terminology~\cite{irs}, electronic currencies nonetheless conform to the conditions we outlined for them to fit within the scope of the model presented here.} When included in the vertex set, these ``currencies'' may fascinatingly not yield a complete graph. In May 2014, for instance, the People's Bank of China began urging banks to be wary of transactions in Bitcoin, a particularly prominent electronic currency. Indeed, China's largest banks banned ``activities related to Bitcoin trading.''~\cite{china} The edge in the market graph between the yuan and Bitcoin is therefore not present. One might further extend the network, data permitting, to include the elicit goods that make up a significant fraction of transactions involving electronic currencies.~\cite{silk} One might expect that the popularity of resorting to alternate payment methods indicates an aversion to trading via traditional centralized currencies, perhaps signaling additional absent edges.

We have so far considered only static market equilibria. Allowing the discussion to remain in the context of currencies, we will briefly discuss considerations for dynamics on exchange networks. Ellerman notes that the no-arbitrage condition of equation \eqref{add} is formally equivalent to Kirchhoff's junction rule of electrodynamics.~\cite{ellerman} We must remember, however, that the economically-relevant system is multiplicative as in \eqref{mult}. It is therefore unrealistic to expect exchange rate dynamics akin to the charging of a circuit in response to changes in consumer demand or central bank supply. Rather, we might imagine a set of basis entries $(e_{i_1,j_1}^*,e_{i_2,j_2}^*,\ldots,e_{i_{n-1},j_{n-1}}^*)$ perturbed according to
\begin{align}
e_{i_k,j_k}^*\mapsto e_{i_k,j_k}^*+\delta e_{i_k,j_k}^*
\end{align}
along with the observation from Definition \ref{exchange} that
\begin{align}
\delta\tilde{E}_{i,j}=\exp(E_{i,j})\delta E.
\end{align}
Letting $A$ be a third-order tensor mapping a basis to its unique arbitrage-free exchange matrix, linearity gives us
\begin{align}
\delta E=A\left(\begin{matrix}\delta e_{i_1,j_1}^*\\\vdots\\\delta e_{i_{n-1},j_{n-1}}^*\end{matrix}\right)
\end{align}
whence we may proceed to solve the dynamic equations. Though Theorem \ref{big} simplifies the calculation by guaranteeing that it is sufficient to know the changes to just $n-1$ independent entries, one must be careful in determining precisely which entries change directly due to a supply or demand perturbation.

Finally, it may be of particular interest to apply these methods on so-called ``scale-free networks,'' those graphs whose degree distributions have power tails. It is well-established that this ubiquitous property arises naturally in economic systems~\cite{fed} and from simple stochastic network formation mechanisms.~\cite{barabasi} Additionally, Miroswski has found success examining arbitrage-free systems with a cellular automaton model~\cite{mirowski07} which, as he points out, are conducive to graph models~\cite[p.~240]{davis} and have been shown to produce asymptotic power tails.~\cite{palasek} Adapting the formulation given here to be specific for such economically- and socially-realistic networks may be a promising line of future inquiry.

\bibliographystyle{unsrt}

\end{multicols}

\end{document}